\documentclass[12pt]{article}

\title{Calibration for Weak Variance-Alpha-Gamma Processes\footnote{This
research was partially supported by ARC grant DP160104037.}}

\author{Boris Buchmann\thanks{Research School of Finance, Actuarial Studies \& Statistics,
Australian National University,
ACT 0200,
Australia.
Email: boris.buchmann@anu.edu.au
}
\and
Kevin W. Lu\thanks{Mathematical Sciences Institute,
	Australian National University,
	ACT 0200,
	Australia.
	Email: kevin.lu@anu.edu.au}
\and
Dilip B. Madan\thanks{Robert H. Smith School of Business, University of Maryland, College Park, MD. 20742, USA, Email: dbm@rhsmith.umd.edu}
}

\usepackage{rotating}
\usepackage{enumerate}
\usepackage{amsmath,latexsym,amssymb,amsthm}
\usepackage{graphicx}	
\usepackage{graphics}											
\usepackage{import}												
\usepackage{graphicx}												
\usepackage{color}													
\usepackage{multirow}
\usepackage{chngpage}
\RequirePackage[colorlinks,citecolor=blue,urlcolor=blue]{hyperref}

\newcommand{\rmd}{{\rm d}}
\newcommand{\rmi}{{\rm i}}

\newcommand{\eqd}{\stackrel{D}{=}}

\newcommand{\EE}{\mathbb{E}}
\newcommand{\DD}{\mathbb{D}}
\newcommand{\RR}{\mathbb{R}}

\newcommand{\CC}{\mathbb{C}}

\newcommand{\PP}{\mathbb{P}}

\newcommand{\TTT}{{\cal T}}

\newcommand{\XXX}{{\cal X}}

\newcommand{\skal}[2]{\left\langle #1,#2\right\rangle}

\newcommand{\eins}{{\bf 1}}

\newcommand{\bfnull}{{\bf 0}}

\newcommand{\bfe}{{\bf e}}
\newcommand{\bfm}{{\bf m}}

\newcommand{\bft}{{\bf t}}
\newcommand{\bfx}{{\bf x}}
\newcommand{\bfy}{{\bf y}}

\newcommand{\bfalpha}{\boldsymbol{\alpha}}

\newcommand{\bfmu}{\boldsymbol{\mu}}

\newcommand{\bftheta}{\boldsymbol{\theta}}

\newcommand{\myCov}{{\rm Cov}}

\newcommand{\tr}{\diamond}

\numberwithin{equation}{section}

\newtheorem{lemma}{Lemma}
\newtheorem{remark0}{\bf Remark}
\newenvironment{remark}{\begin{remark0}\em}{\end{remark0}\par}

\newtheorem{proposition}{Proposition}
\numberwithin{theorem}{section}
\numberwithin{proposition}{section}
\numberwithin{lemma}{section}
\numberwithin{corollary}{section}
\numberwithin{remark0}{section}
\numberwithin{definition}{section}

\begin{document}

\maketitle

\begin{abstract}
The weak variance-alpha-gamma process is a multivariate L\'evy process constructed by weakly subordinating Brownian motion, possibly with correlated components with an alpha-gamma subordinator. It generalises the variance-alpha-gamma process of Semeraro constructed by traditional subordination. We compare three calibration methods for the weak variance-alpha-gamma process, method of moments, maximum likelihood estimation (MLE) and digital moment estimation (DME). We derive a condition for Fourier invertibility needed to apply MLE and show in our simulations that MLE produces a better fit when this condition holds, while DME produces a better fit when it is violated. We also find that the weak variance-alpha-gamma process exhibits a wider range of dependence and produces a significantly better fit than the variance-alpha-gamma process on an S\&P500-FTSE100 data set, and that DME produces the best fit in this situation.
\end{abstract}

\noindent {\em 2000 MSC Subject Classifications:} \ Primary: 60G51\\
\ Secondary: 62F10, 60E10\\
\noindent {\em Keywords:} Brownian Motion, Gamma Process, L\'evy Process, Subordination, Variance-Gamma, Variance-Alpha-Gamma, Self-Decomposability, Log-Return, Method of Moments, Maximum Likelihood Estimation, Digital Moment Estimation.


\section{Introduction}\label{secintro}

The subordination of Brownian motion has important applications in mathematical finance, and acts as a time change that models the flow of information, measuring time in volume of trade, as opposed to real time. This idea was initiated by Madan and Seneta in~\cite{MaSe90} who introduced the variance-gamma ($VG$) process for modelling stock prices, where the subordinate is Brownian motion and the subordinator is a gamma process.

Subordination can be applied to model dependence in multivariate price processes. The multivariate $VG$ process in~\cite{MaSe90} uses $n$-dimensional Brownian motion as its subordinate and a univariate gamma process as its subordinator, which gives it a restrictive dependence structure, where components cannot have idiosyncratic time changes and must have equal kurtosis when there is no skewness. Models based on linear combinations of independent L\'evy processes \cite{Ka09,Ma11} also do not account for both common and idiosyncratic time changes. These deficiencies are addressed by the use of an alpha-gamma subordinator, resulting in the variance-alpha-gamma ($VAG$) process which was introduced by Semeraro in~\cite{Se08} and also studied in~\cite{Gu13,LS10}. However, in this case, the Brownian motion subordinate must have {\em independent} components, which also restricts the dependence structure.

To be precise, let $B=(B_1,\dots,B_n)$, where $T=(T_1,\dots,T_n)$ be independent $n$-dimensi\-onal processes, where $B$ is Brownian motion and $T$ is a subordinator. Subordination is the operation that produces the process $B\circ T$ defined by $(B\circ T)(t):=(B_1(T_1(t)),\dots,B_n(T_n(t))),t\ge 0$. Subordination in the case when $T$ has indistinguishable components has been studied in \cite{BS10,s}, and when $B$ has independent components in \cite{BPS01}. In these cases, which we refer to as traditional subordination, $B\circ T$ is a L\'evy process, otherwise it may not be~(see~\cite{BLMa}, their Proposition~3.9). We refer the reader to~\cite{BKMS16} for a thorough discussion of traditional subordination and its applications.

In \cite{BLMa}, we introduced the weak subordination of $B$ and $T$, an operation that extends traditional subordination and always produces a L\'evy process $B\odot T$. Then the weak variance-alpha-gamma ($WVAG$) process can be constructed using weak subordination instead of traditional subordination, while allowing for the Brownian motion to have possibly correlated components. The $WVAG $ process exhibits a wider range of dependence while remaining parsimoniously parametrised, each component has both common and idiosyncratic time changes, it has $VG$ marginals with independent levels of kurtosis, and the jump measure has full support.

Weak subordination also has applied in quantitative finance. In \cite{MiSz17}, various marginal consistent dependence models have been constructed by weak subordination. In \cite{Ma18}, log return modelling based on the $WVAG$ process was applied in instantaneous portfolio theory. In \cite{Mi18}, weak subordination using subordinators with arbitrary marginal components and dependence specified by a L\'{e}vy copula was studied in the context of financial information flows.

Maximum likelihood estimation (MLE) has been used to fit financial data to a univariate $VG$ process in \cite{MCC98,FiS}, to a bivariate $VG$ process in \cite{FS07}, to a $WVAG$ process in \cite{MiSz17}, and to a factor-based subordinated Brownian motion in \cite{LMS16,MiSz17,Wa09}, a generalisation of the $WVAG$ process. Since the density function of the $VAG$ and $WVAG$ distribution is not explicitly known but its characteristic function is, the density function is computed using Fourier inversion.

In this paper, we derive a sufficient condition in terms of the parameters for Fourier invertibility, a problem that to our knowledge is not addressed in the existing literature. Then we compare MLE with method of moments (MOM) and digital moment estimation (DME) from \cite{Ma15}. Using simulations we find that MLE produces a better fit when the Fourier invertibility condition is satisfied but that DME is better when it is violated. In addition, we fit both the $WVAG $ and $VAG $ model to an S\&P500-FTSE100 data set and show that the weak model has a significantly better fit, and that DME is the better method in this situation. Finally, using a condition for the self-decomposability of the $WVAG$ process from \cite{BLMb}, we find that the log returns are self-decomposable.

This paper is structured as follows. In Section 2, we review the definition and properties of the $WVAG $ process, and other preliminaries. In Section 3, we derive a condition for Fourier invertibility. In Section 4, we apply MOM, MLE, DME to simulated and real data, and discuss our findings. In Section 5, we conclude the paper.


\section{Weak Variance-Alpha-Gamma Process}
Let $\RR^n$ be $n$-dimensional Euclidean space whose elements are row vectors $\bfx=(x_1,\dots,x_n)$ with canonical basis $\{\bfe_k\!:\!1\!\le\!k\!\le\!n\}$.
Let $\skal \bfx \bfy=\bfx \bfy'$ denote the Euclidean product, $\|\bfx\|^2=\bfx\bfx'$ denote the Euclidean norm, and let $\|\bfx\|^2_\Sigma:=\bfx\Sigma\bfx'$.
For $n$-dimensional processes $X$ and $Y$, $X\eqd Y$ indicates that $X$ and $Y$ are identical in law, that is their systems of finite dimensional distributions are equal.

A overview of L\'evy processes and weak subordination is given in the appendix. Throughout, $B=(B_1,\dots,B_n)\sim BM^n(\bfmu,\Sigma)$ refers to an $n$-dimensional Brownian motion with linear drift $\EE[B(t)]=\bfmu t$ and covariance matrix Cov$(B(t))=t\Sigma$, $t\geq0$.
An $n$-dimensional {\em subordinator} $T=(T_1,\dots,T_n)\sim S^n(\TTT)$ is an $n$-dimensional L\'evy process with nondecreasing components, and its L\'evy measure is denoted by $\TTT$.\\[1mm]
\noindent{\bf Gamma subordinator.}~For $a,b>0$, a univariate subordinator $G\sim\Gamma_S(a,b)$ is a {\em gamma subordinator} if its marginal $G(t)$, $t\!\geq\!0$,
is gamma distributed with shape parameter $at$ and rate parameter $b$. 
If $a\!=\!b$, we refer to $G$ as a standard gamma subordinator, in short, $G\sim\Gamma_S(b):=\Gamma_S(b,b)$.\\[1mm]
\noindent{\bf Alpha-gamma subordinator.}~Assume $n\!\ge\!2$. Let $\bfalpha\!=\!(\alpha_1,\dots,\alpha_n)\!\in\!(0,\infty)^n$ and $G_0,\dots,G_{n}$ be independent gamma subordinators such that $G_{0}\!\sim\!\Gamma_S(a,1)$, $G_k\!\sim\!\Gamma_S(\beta_k,1/\alpha_k)$, where $a\!>\!0$, $a\alpha_k\!<\!1$, $\beta_k\!:=\!(1-a\alpha_k)/\alpha_k$, $1\!\le\! k\!\le\! n$. A process $T\sim AG_S^n(a,\bfalpha)$ is an {\em alpha-gamma ($AG_S^n$) subordinator}~\cite{Se08} with parameters $a,\bfalpha$ if $T\eqd G_0 \bfalpha+(G_1,\dots,G_n)$. An alpha-gamma subordinator $T$ has correlated components with marginals $T_k\sim\Gamma_S(1/\alpha_k)$, $1\le k\le n$.\\[1mm]
\noindent{\bf Variance-gamma process.}~Let $b\!>\!0$, $\bfmu\!\in\!\RR^n$ and $\Sigma\!\in\!\RR^{n\times n}$ be a covariance matrix. A process $V\!\sim\! VG^n(b,\bfmu,\Sigma)$ is a {\em variance-gamma $(VG^n)$ process}~\cite{MaSe90} with parameters $b,\bfmu,\Sigma$ if $V\!\sim\! BM^n(\bfmu,\Sigma)\circ (\Gamma_S(b)\bfe)$, where $\bfe:=(1,\dots,1)\in\RR^n$.

The characteristic exponent of $V$ is (see~\cite{BKMS16}, their Formula~(2.9))
\begin{equation}\label{charexpoVG}
\Psi_V(\bftheta)=-b\ln\left\lbrace 1-\frac{\rmi\skal{\bfmu}{\bftheta}}{b}+\frac{\|\bftheta\|^2_\Sigma}{2b}\right\rbrace \,,\quad \bftheta\in\RR^n\,,
\end{equation}
where $\ln:\CC\backslash(-\infty,0]\to\CC$ is the principal branch of the logarithm.\\[1mm]
\noindent{\bf Strong variance-alpha-gamma process.}~Assume $n\!\ge\!2$. Let $\bfmu\!\in\!\RR^n$ and $\Sigma\in\![0,\infty)^{n\times n}$ be a {\em diagonal} matrix. A process $X\sim VAG^n(a,\bfalpha,\bfmu,\Sigma)$ is a {\em (strong) variance-alpha-gamma ($VAG^n$) process}~\cite{LS10,Se08} with parameters $a,\bfalpha,\bfmu,\Sigma$ if $X\sim BM^n(\bfmu,\Sigma)\circ AG^n_S(a,\bfalpha)$.

In \cite{BLMa}, the weak $VAG$ process was formulate using weak subordination, allowing $B$ to have dependent components while remaining a L\'evy process.\\[1mm]
\noindent{\bf Weak variance-alpha-gamma process.}~Assume $n\!\ge\!2$. Let $\bfmu\!\in\!\RR^n$ and $\Sigma\!\in\!\RR^{n\times n}$ be an {\em arbitrary} covariance matrix. The process $X\sim WVAG^n(a,\bfalpha,\allowbreak\bfmu,\Sigma)$ is a \emph{weak variance-alpha-gamma process}~\cite{BLMa} with parameters $a,\bfalpha,\bfmu,\Sigma$ if $X\sim BM^n(\bfmu,\Sigma)\odot AG^n_S(a,\bfalpha)$, where $\odot$ denotes the weak subordination operation (see the appendix).

Next, we gather various known results about the $WVAG$ process that will be useful later on. The notation $\tr$ is defined in \eqref{outproddefn} and self-decomposability is defined in the appendix.
\begin{proposition}\label{wvagprop}
	Let $n\geq2$ and $X\sim WVAG^n(a,\bfalpha,\bfmu,\Sigma)$.
	\begin{enumerate}
		\item[(i)] $X$ is an $n$-dimensional L\'evy process with L\'evy exponent
		\begin{align} \Psi_X(\bftheta)=&-a\;\ln\left\lbrace 1-\rmi \skal{\bfalpha\tr\bfmu}\bftheta+\frac{1}{2}\|\bftheta\|^2_{\bfalpha\tr\Sigma}\right\rbrace \nonumber\\
		&-\sum_{k=1}^n\beta_k\ln\left\lbrace 1-\rmi \alpha_k\mu_k\theta_k+\frac{1}{2}\alpha_k\theta_k^2\Sigma_{kk}\right\rbrace,\quad\bftheta\in\RR^n \,.\label{v-alpha-g-cf}
		\end{align}
		\item[(ii)] Let $V_0\sim VG^n\left( a,a\bfalpha\tr\bfmu,a\bfalpha\tr\Sigma\right)$, $V_k\sim VG^1(\beta_k,(1-a\alpha_k)\mu_k,(1\!-\!a\alpha_k)\allowbreak\Sigma_{kk})$, $1\!\le\! k\!\le\! n$ be independent. Then $X\eqd V_0+\sum_{k=1}^nV_k\bfe_k$.
		\item[(iii)]For any $c>0$, $(X(ct))_{t\geq0}\sim WVAG^n(ca,\bfalpha/c,c\bfmu,c\Sigma)$.
		\item[(iv)] For $1\!\le\! k\!\le\! n$, $X$ has marginal distribution $X_k\sim VG^1(1/\alpha_k,\mu_k,\Sigma_{kk})$.
		\item[(v)] If $\Sigma$ is diagonal, then $X\sim VAG^n(a,\bfalpha,\bfmu,\Sigma)$.
		\item[(vi)] For $1\!\le\! k\!\neq\!l\!\le\! n$, $\myCov(X_k(1),X_l(1))=a(\alpha_k\!\wedge\!\alpha_l)\Sigma_{kl}+a\alpha_k\alpha_l\mu_k\mu_l$.
		\item[(vii)] If $\Sigma$ is invertible, then $X$ is self-decomposable if and only if $\bfmu={\bf 0}$.
	\end{enumerate}
\end{proposition}
\begin{proof}
See \cite{BLMa} for (i), (ii), (iv)-(vi), \cite{MiSz17} for (iii), and \cite{BLMb} for (vii).
\end{proof}
The $WVAG $ process exhibits a wider range of dependence than the $VAG $ process. For example, it has an additional covariance term $a(\alpha_k\!\wedge\!\alpha_l)\Sigma_{kl}$ from Proposition \ref{wvagprop} (vi).

\section{Fourier Invertibility}

Let $n\geq2$ and $X\sim WVAG ^n(a,\bfalpha,\bfmu,\Sigma)$ and $Y\!=\!I\bfm + X$, $\bfm\in\RR^n$, where $I:[0,\infty)\to[0,\infty)$ is the identity function. The density function of $Y(t)$, $t>0$, which is needed for MLE, exists because $Y(t)$ has an absolutely continuous distribution for. However, it is not explicitly known, so it is computed using Fourier inversion as
\begin{align}\label{fourierinvert}
f_{Y(t)}(\bfy) = (2\pi)^{-n}\int_{\mathbb{R}^n}\exp(-\rmi \skal{\bftheta}{\bfy-\mathbf{m}})\Phi_{X(t)}(\bftheta)\,\rmd\bftheta,\quad \bfy\in\RR^n\,,
\end{align}
where $\Phi_{X(t)}(\bftheta)=\exp(t\Psi_X(\bftheta))$, $\bftheta\in\RR^n$, and $\Psi_X(\bftheta)$ from \eqref{v-alpha-g-cf}, provided $\Phi_{X(t)}\in L^1$. If $\Phi_{X(t)}\in L^1$, we say that $X(t)$ is Fourier invertible and we give a condition for this in terms of an inequality relating the parameters.

\begin{lemma}\label{lemma1}
	Let $B\sim BM^n(\bfmu,\Sigma)$, $B^* \sim BM^n({\bf 0},\Sigma)$, $T \sim S^n(\bfnull,\TTT)$, $X \eqd B \odot T$, $X^* \eqd B^*\odot T$, $Y\eqd I\bfm+X$, $\bfm \in \RR^n$. For all $t\geq0$ and $p>0$, if $\Phi_{X^*(t)}\in L^p$, then $\Phi_{Y(t)}\in L^p$.
\end{lemma}
\begin{proof} 
	For all $t\geq0$, $\Phi_{Y(t)}(\bftheta)=e^{\rmi t\skal{\bftheta}{\bfm}}\Phi_{X(t)}(\bftheta)$, $\bftheta\in\RR^n$, so that $|\Phi_{Y(t)}(\bftheta)|\allowbreak=\exp(t\Re \Psi_{X}(\bftheta))$. Using \eqref{wsexp}, we have
	\begin{align*}
		\Re \Psi_{X}(\bftheta)&=\int_{[0,\infty)_*^n} (\Re\Phi_{B(\bft)}(\bftheta)-1)\,\TTT(\rmd \bft)\\
		&\le \int_{[0,\infty)_*^n} (|\Phi_{B(\bft)}(\bftheta)|-1)\,\TTT(\rmd \bft)\\
		&=\int_{[0,\infty)_*^n} \left( \exp\left(-\frac{1}{2}\|\bftheta\|^2_{\bft\tr\Sigma}\right)-1\right) \,\TTT(\rmd \bft)\\
		&=\Re \Psi_{X^*}(\bftheta).
	\end{align*}
	Therefore, $|\Phi_{Y(t)}(\bftheta)|\le|\Phi_{X^*(t)}(\bftheta)|$, from which the result follows.
\end{proof}

\begin{lemma}\label{fourierinvpropvg}
	Let $V\sim VG^n(b,\bfmu,\Sigma)$, and assume that $\Sigma$ is invertible. Let $p>0$. If $pb>n/2$, then $\Phi_{V}\in L^p$.
\end{lemma}

\begin{proof} Since variance-gamma processes are weakly subordinated processes (see \eqref{strongequalsweak}), we can apply Lemma~\ref{lemma1}, which means that we can assume $\bfmu=\bfnull$. For $V\sim VG^n(b,0,\Sigma)$, by \eqref{charexpoVG}, $V$ has characteristic function
	\begin{align*}
		\Phi_{V}(\bftheta)=\left(1+\frac{\left\| \bftheta\right\|_{\Sigma}^2}{2b}\right)^{-b},\quad\bftheta\in\RR^n.
	\end{align*}
	Using the Cholesky decomposition, $\Sigma=U'U$, where $U$ is a lower triangular matrix with positive elements on the diagonal. Let $p>0$. Making the transformation $\bftheta=(2b)^{1/2}\bfx (U')^{-1}$, noting that $(U')^{-1}$ exists, and hence the transformation is injective, we have
	\begin{align}\label{vgfourierint}
		\int_{\RR^n} |\Phi_{V}(\bftheta)|^p\,\rmd\bftheta=|(2b)^{1/2}U^{-1}|\int_{\RR^n} \left( 1+\left\| \bfx\right\|^2\right) ^{-pb} \,\rmd\bfx.
	\end{align}
	Using the polar decomposition (see Corollary~B.7.7 in \cite{sas13}) on the RHS of \eqref{vgfourierint}, we have $\Phi_{V}\in L^p$ if and only if
	\begin{align*}
		\int_0^\infty (1+r^2)^{-pb}r^{n-1}\,\rmd r<\infty,
	\end{align*}
	which is equivalent to $pb>n/2$.
\end{proof}


\begin{proposition}\label{fourierinvprop}
	Let $X\sim WVAG ^n(a,\bfalpha,\bfmu,\Sigma)$ and $Y \eqd I\bfm + X$, $\bfm\in\RR^n$. Assume that $\Sigma$ is invertible. For $t>0$, if
	\begin{align}\label{cond}
		\left(\frac{a}{n}+\min_{1\le k \le n}\beta_k\right)t>\frac{1}{2},
	\end{align}
	then $\Phi_{X(t)},\Phi_{Y(t)}\in L^1$.
\end{proposition}

\begin{proof} By Proposition~\ref{wvagprop} (iii), it suffices to prove the result for $t=1$, and by Lemma~\ref{lemma1}, we can assume $\bfmu=\bfnull$ and $\bfm=\bfnull$, so that $Y\sim WVAG ^n(a,\bfalpha,{\bf 0},\Sigma)$.
	
	Let $V_0\sim VG^n(a,\mathbf{0},\allowbreak a \bfalpha\tr\Sigma)$, $V_k\sim VG^1\left(\beta_k,0,(1-a\alpha_k)\Sigma_{kk}\right)$, $1 \leq k \leq n$, be independent, and let $V^*:=(V_1,\dots,V_n)$. By Proposition~\ref{wvagprop} (ii), $Y$ has characteristic function $\Phi_Y(\bftheta)=\Phi_{V_0}(\bftheta)\Phi_{V^*}(\bftheta)$, where $\Phi_{V^*}(\bftheta):=\prod_{k=1}^n\Phi_{V_k}(\theta_k)$. For $p^{-1}+q^{-1}=1$, $p,q>1$, 
	H\"{o}lder's inequality gives
	\begin{align*}
		\int_{\RR^n}|\Phi_Y(\bftheta)|\,\rmd\bftheta&\le \left( \int_{\RR^n}|\Phi_{V_0}(\bftheta)|^{p}\,\rmd\bftheta\right)^{1/p} \left( \int_{\RR^n}|\Phi_{V^*}(\bftheta)|^{q}\,\rmd\bftheta\right)^{1/q}\\
		&= \left( \int_{\RR^n}|\Phi_{V_0}(\bftheta)|^{p}\,\rmd\bftheta\right)^{1/p} \prod_{k=1}^n\left( \int_{\RR}|\Phi_{V_k}(\theta)|^{q}\,\rmd\theta\right)^{1/q}.
	\end{align*}
	By Lemma~\ref{fourierinvpropvg}, this integral is finite when $pa>n/2$, $q\beta_k>1/2$ and $p,q>1$ for all $1\le k \le n$. Thus,
	\begin{align*}
		1=\frac{1}{p}+\frac{1}{q}<2\left(\frac{a}{n}\wedge\frac{1}{2}\right) +2\left( \min_{1\le k \le n}\beta_k\wedge\frac 12\right) ,
	\end{align*}
	which is equivalent to \eqref{cond}.
\end{proof}

\begin{remark} Note that this condition for a $VG^1(b,\mu,\Sigma)$ distribution to be Fourier invertible is identical to the condition for its density function having no singularity in \cite{KuTa}, which is $b>1/2$.
\end{remark}

\begin{remark} We see that for sufficiently small $t>0$, \eqref{cond} will not be satisfied. This means that using \eqref{fourierinvert} to compute the density function may not be valid when attempting parameter estimation for a $WVAG $ process based on observations from such a sufficiently small sampling interval.
\end{remark}
\section{Calibration}

We now specialise to the case of $n=2$. Let $X\sim WVAG ^2(a,\bfalpha,\bfmu,\Sigma)$, $Y:=(Y_1,Y_2)=I\bfm +X$, $\mathbf{m}\in\RR^2$. Let $(S_1, S_2)$ be a bivariate price process
\begin{align}\label{priceproc}
S_k(t)=S_k(0)\exp(Y_k(t)),\quad t\ge0,\quad k=1,2.
\end{align}
For $N$ equally spaced discrete observations with sampling interval $c>0$, the log returns are
\begin{eqnarray*}
	\mathbf{y}_j:=(y_{1j},y_{2j}):=\Big(\ln\frac{S_1(jc)}{S_1((j-1)c)},
	\ln\frac{S_2(jc)}{S_2((j-1)c)}\Big)\stackrel{D}{=} Y(c),\quad j=1,\dots, N,
\end{eqnarray*}
and are iid. We call this the {\em $WVAG $ model}. If $\Sigma_{12}=0$, we called it the {\em $VAG $ model} as $X$ reduces to a $VAG $ process by Proposition \ref{wvagprop} (v).

\subsection{Simulation method}

The result in Proposition \ref{wvagprop} (ii) can be used to simulate $X\sim WVAG ^2(a,\bfalpha,\allowbreak\bfmu,\Sigma)$ in terms of $VG^n$ and $VG^1$ processes.

For the sampling intervals $c=1,0.1$ and sample size $N=1000$, we make 100 simulations of $Y$, and estimate the parameters from the observations $(\mathbf{y}_j)_{j=1}^N$ with true parameters $a=1$, $\bfalpha=(0.8,0.6)$, $\bfmu=(0.1,-0.3)$, $\Sigma=[1,0.6;0.6,1.2]$, $\bfm=(-0.1,0.3)$.

\subsection{Calibration methods}\label{parestsec}

We estimate the parameters $(a,\bfalpha,\bfmu,\Sigma,\bfm)$ from the observations $(\mathbf{y}_j)_{j=1}^N$ using method of moments (MOM), which is quick and easy to implement, maximum likelihood estimation (MLE) from Micha\-elsen \& Szimayer~\cite{MiSz17}, which may be expected as being asymptotically optimal under the model, and a modification of digital moment estimation (DME) from Madan \cite{Ma15}, which is more robust to model misspecification.

\noindent{\bf Method of moments.} The initial values of $\mu_k,\alpha_k,\Sigma_{kk},\allowbreak m_k$, $k=1,2$, are obtained by least squares on the first four central moments $\EE(Y_k(c))$, $\EE((Y_k(c)-\EE(Y_k(c)))^p)$, $p=2,3,4$, with the corresponding sample moments. The initial values of the joint parameters $a, \Sigma_{12}$ are obtained by least squares on $\EE((Y_1(c)-\EE(Y_1(c)))^p(Y_2(c)\allowbreak-\EE(Y_2(c)))^p)$, $p=1,2$, with the corresponding sample moments, with $p=1$ excluded when fitting the $VAG $ model. Using these initial values, least squares is solved over all parameters. Note that this last step has no effect when these moments can be matched exactly.

Moment formulas can be found in \cite{MiSz17}.\\[1mm]
\noindent{\bf Maximum likelihood estimation.}
The density function of $Y(c)$ is not explicitly known so it is numerically computed using Fourier inversion by \eqref{fourierinvert}. The numerical optimisation needed to implement MLE requires initial values. The first initial values can be obtained by MOM. Using the first initial values, MLE is applied to each marginal observations to obtain the second initial values of $\mu_k,\alpha_k,\Sigma_{kk},\allowbreak m_k$, $k=1,2$, and to the bivariate observations to obtain the second initial values of $a, \Sigma_{12}$. Finally, using the second initial values, MLE is applied on all parameters. For the $VAG $ model, we apply the above method with the constraint $\Sigma_{12}=0$.\\[1mm]
\noindent{\bf Digital moment estimation.}
Let $k=1,2$, let $\mathbf{q}$ be the vector of 10 equally spaced points from 0.05 to 0.95, and let $\mathbb{P}_k$ be the empirical quantiles of the observations $(y_{kj})_{j=1}^N$ at the probabilities $\mathbf{q}$. Let $p_{y}(\mu_k,\alpha_k,\Sigma_{kk},m_k) := \PP(Y_k(c)\leq y)$, where $y\in \mathbb{P}_k$, $Y_k\sim m_kI+VG^1(1/\alpha_k,\mu_k,\Sigma_{kk})$ (see Proposition \ref{wvagprop} (iv)), and $q_{y}$ is the corresponding empirical probability. Marginal parameters $\mu_k,\alpha_k,\Sigma_{kk},\allowbreak m_k$ are estimated by minimizing the error $\sum_{y\in\mathbb{P}_k} (p_y(\mu_k,\alpha_k,\allowbreak\Sigma_{kk},m_k)-q_y)^2$.
		
With the estimated marginal parameters, let $\rho:=\Sigma_{12}\big/(\Sigma_{11}\Sigma_{22})^{1/2}$, $p_{\mathbf{y}}(a,\allowbreak\rho):= \PP(Y_1(c)\leq y_1, Y_2(c)\leq y_2)$, where $\mathbf{y}\!:=\!(y_1,y_2)\!\in\!\mathbb{P}_1\!\times\!\mathbb{P}_2$, $Y\sim I\bfm+WVAG ^2(a,\bfalpha,\bfmu,\Sigma)$, and $q_{\mathbf{y}}$ is the corresponding empirical probability. Since $p_{\mathbf{y}}(a,\rho)$ is computationally expensive to calculate directly, it is estimated by the empirical probability over 10000 simulations. The joint parameters $a$, $\rho$ are estimated by minimizing the LOESS smooth \cite{CGS} of the error $\sum_{\mathbf{y}\in\mathbb{P}_1\times\mathbb{P}_2} (p_{\mathbf{y}}(a,\rho)-q_{\mathbf{y}})^2$. The predictor variables for the LOESS smooth are 100 equally spaced points on the feasible set of $(a,\rho)\!\in\!(0,(1/\alpha_1)\!\wedge\!(1/\alpha_2))\!\times\!(-1,1)$. For the $VAG $ model, we apply the above method with the constraint $\rho=0$.

\subsection{Goodness of fit statistics}

To assess the overall goodness of fit of each parameter estimation method, as opposed to assessing individual parameters, we consider 3 goodness of fit statistics, the negative log-likelihood ($-\log L$), a chi-squared ($\chi^2$) statistic, and a Kolmogorov-Smirnov (KS) statistic.

To compute $\chi^2$, we apply the Rosenblatt transform \cite{R52} of the fitted distribution to the observations, which has a uniform distribution on $[0,1]^2$ if the fitted distribution coincides with the true distribution, and then we compute the $\chi^2$ statistic for a test of uniformity over an equally spaced partition of $[0,1]^2$ into 100 cells. Since computing $-\log L$ and $\chi^2$ requires Fourier inversion, it may not be possible to compute these statistics accurately when the Fourier invertibility condition does not hold, so they are not displayed in Table~\ref{v-alpha-g-table2}.

Therefore, we also consider the 2-dimensional, two-sample Kolmogorov-Smirnov statistic introduce by Peacock in \cite{peac}, and computed using the method of \cite{xiao}. This is the statistic for testing equality of the fitted distribution and the true distribution based on a sample from the respective distributions, and therefore does not require the density function $f_{Y(c)}(\bfy)$ or Fourier inversion. When applied to real data in Subsection~\ref{realdatasec}, we take the average of the KS statistics computed from the observations and 100 samples from the fitted distribution. When applied to simulated data in Subsection~\ref{simsec}, the KS statistic is computed from the observations and a sample from the fitted distribution. All 3 goodness of fit statistics were averaged over the 100 simulations.

\subsection{Quantile choice for DME}

Different choices of quantiles for DME are possible. Let $\mathbf{q}_1$ be the vector of 10 equally spaced points from 0.05 to 0.95, $\mathbf{q}_2$ be the vector of 10 equally spaced points from 0.01 to 0.99, $\mathbf{q}_3$ be the vector of 10 equally spaced points from 0.1 to 0.9, $\mathbf{q}_4$ be the vector of 20 equally spaced points from 0.05 to 0.95. For sampling interval $c=1$, Table~\ref{qunatilechoice} shows the goodness of fit for 4 choices of quantiles. We find $\mathbf{q}=\mathbf{q}_1$ as yielding the lowest RMSE for most variables and the lowest goodness of fit statistics. However, given that the results are so similar, these quantile choices make only a small difference to the overall goodness of fit.
\begin{table} 
	\begin{center}
		\begin{tabular}{cccccc}
			\hline
			{\bf Parameter} & {\bf True value} & {\bf $\mathbf{q}_1$} & {\bf $\mathbf{q}_2$} & {\bf $\mathbf{q}_3$} & {\bf $\mathbf{q}_4$} \\ \hline
			$a$                 & $\phantom{-}1$    & $0.171$    & $0.171$    & $0.182$    & $0.175$ \\
			$\alpha_1$          & $\phantom{-}0.8$  & $0.127$    & $0.132$    & $0.143$    & $0.128$ \\
			$\alpha_2$          & $\phantom{-}0.6$  & $0.126$    & $0.145$    & $0.149$    & $0.129$ \\
			$\mu_1$             & $\phantom{-}0.1$  & $0.062$    & $0.066$    & $0.066$    & $0.062$ \\
			$\mu_2$             & $-0.3$            & $0.121$    & $0.271$    & $0.229$    & $0.188$ \\
			$\Sigma_{11}$       & $\phantom{-}1$    & $0.084$    & $0.083$    & $0.093$    & $0.084$ \\
			$\Sigma_{22}$       & $\phantom{-}1.2$  & $0.113$    & $0.166$    & $0.147$    & $0.123$ \\
			$\Sigma_{12}$       & $\phantom{-}0.6$  & $0.154$    & $0.182$    & $0.172$    & $0.150$ \\
			$m_1$               & $-0.1$            & $0.051$    & $0.054$    & $0.053$    & $0.050$ \\
			$m_2$               & $\phantom{-}0.3$  & $0.110$    & $0.262$    & $0.219$    & $0.179$ \\ \hline
			$-\ln L$            &                   & $2791.674$ & $2795.826$ & $2794.374$ & $2792.303$ \\
			$\chi^2$            &                   & $93.848$   & $97.292$   & $96.728$   & $95.078$   \\
			KS                  &                   & $0.054$    & $0.055$    & $0.054$    & $0.054$    \\ \hline
		\end{tabular}
		\caption{RMSE using DME with quantiles $\mathbf{q}_1,\dots,\mathbf{q}_4$ for the $WV AG$ model fitted to simulated data with $c=1$.}
		\label{qunatilechoice}
	\end{center}
\end{table}

\subsection{Simulated data results}\label{simsec}

For the sampling interval $c=1$, the Fourier invertibility condition is satisfied as the LHS of \eqref{cond} is $0.75>1/2$. The calibration results for the $WVAG$ model with $c=1$ is shown in Table \ref{v-alpha-g-table2}. Here, we find that MLE gives the best fit with the lowest $\chi^2$ statistic. The KS statistic for MLE and DME are approximately equal.

\begin{table} 
	\begin{adjustwidth}{-3cm}{-3cm}
	\begin{center}
		\begin{tabular}{ccccc}
			\hline
			{\bf Parameter} & {\bf True value} & {\bf MOM} & {\bf MLE} & {\bf DME} \\ \hline
			$a$                 & $\phantom{-}1$    & $\phantom{-}0.920$ $(0.424)$ & $\phantom{-}0.983$ $(0.242)$ & $\phantom{-}0.902$ $(0.171)$ \\
			$\alpha_1$          & $\phantom{-}0.8$  & $\phantom{-}0.806$ $(0.342)$ & $\phantom{-}0.824$ $(0.111)$ & $\phantom{-}0.818$ $(0.127)$ \\
			$\alpha_2$          & $\phantom{-}0.6$  & $\phantom{-}0.589$ $(0.216)$ & $\phantom{-}0.594$ $(0.094)$ & $\phantom{-}0.589$ $(0.126)$ \\
			$\mu_1$             & $\phantom{-}0.1$  & $\phantom{-}0.103$ $(0.097)$ & $\phantom{-}0.103$ $(0.053)$ & $\phantom{-}0.096$ $(0.062)$ \\
			$\mu_2$             & $-0.3$            & $-0.310$ $(0.131)$           & $-0.301$ $(0.083)$           & $-0.313$ $(0.121)$           \\
			$\Sigma_{11}$       & $\phantom{-}1$    & $\phantom{-}0.989$ $(0.078)$ & $\phantom{-}1.006$ $(0.071)$ & $\phantom{-}0.993$ $(0.084)$ \\
			$\Sigma_{22}$       & $\phantom{-}1.2$  & $\phantom{-}1.177$ $(0.088)$ & $\phantom{-}1.202$ $(0.086)$ & $\phantom{-}1.179$ $(0.113)$ \\
			$\Sigma_{12}$       & $\phantom{-}0.6$  & $\phantom{-}0.835$ $(0.335)$ & $\phantom{-}0.669$ $(0.192)$ & $\phantom{-}0.639$ $(0.154)$ \\
			$m_1$               & $-0.1$            & $-0.103$ $(0.089)$           & $-0.105$ $(0.045)$           & $-0.097$ $(0.051)$           \\
			$m_2$               & $\phantom{-}0.3$  & $\phantom{-}0.313$ $(0.120)$ & $\phantom{-}0.302$ $(0.070)$ & $\phantom{-}0.314$ $(0.110)$ \\ \hline
			$-\ln L$            &                   & $\phantom{-}2802.337$        & $\phantom{-}2787.513$        & $\phantom{-}2791.674$        \\
			$\chi^2$            &                   & $\phantom{-}119.052$         & $\phantom{-}91.268$          & $\phantom{-}93.848$          \\
			KS                  &                   & $\phantom{-}0.068$           & $\phantom{-}0.054$           & $\phantom{-}0.054$           \\ \hline
		\end{tabular}
	\end{center}
	\end{adjustwidth}
	\caption{Expected value of estimates and RMSE (in parentheses) for the $WVAG $ model fitted to simulated data with $c=1$.}
	\label{v-alpha-g-table2}
\end{table}

For the sampling interval $c=0.1$, the Fourier invertibility condition is violated as the LHS of \eqref{cond} is $0.08<1/2$. The corresponding results are shown in Table \ref{v-alpha-g-table4}. Here, we find that DME gives the best fit with the lowest KS statistic, however MLE still produces a good fit and does not break down. This suggests that the condition may not be necessary for the MLE to produce accurate parameter estimates. In both cases, $c=1,0.1$, the RMSE and goodness of fit statistics are highest for MOM.

\begin{table}[!htbp]
	\begin{adjustwidth}{-3cm}{-3cm}
	\begin{center}
		\begin{tabular}{ccccc}
			\hline
			{\bf Parameter} & {\bf True value} & {\bf MOM} & {\bf MLE} & {\bf DME} \\ \hline
			$a$                 & $\phantom{-}1$    & $\phantom{-}1.106$ $(0.507)$ & $\phantom{-}0.990$ $(0.062)$ & $\phantom{-}0.896$ $(0.121)$ \\
			$\alpha_1$          & $\phantom{-}0.8$  & $\phantom{-}0.636$ $(0.247)$ & $\phantom{-}0.782$ $(0.033)$ & $\phantom{-}0.796$ $(0.057)$ \\
			$\alpha_2$          & $\phantom{-}0.6$  & $\phantom{-}0.504$ $(0.198)$ & $\phantom{-}0.602$ $(0.026)$ & $\phantom{-}0.603$ $(0.031)$ \\
			$\mu_1$             & $\phantom{-}0.1$  & $\phantom{-}0.099$ $(0.167)$ & $\phantom{-}0.114$ $(0.099)$ & $\phantom{-}0.104$ $(0.170)$ \\
			$\mu_2$             & $-0.3$            & $-0.347$ $(0.219)$           & $-0.250$ $(0.123)$           & $-0.301$ $(0.146)$           \\
			$\Sigma_{11}$       & $\phantom{-}1$    & $\phantom{-}0.992$ $(0.136)$ & $\phantom{-}1.005$ $(0.133)$ & $\phantom{-}1.013$ $(0.302)$ \\
			$\Sigma_{22}$       & $\phantom{-}1.2$  & $\phantom{-}1.197$ $(0.166)$ & $\phantom{-}1.245$ $(0.161)$ & $\phantom{-}1.234$ $(0.221)$ \\
			$\Sigma_{12}$       & $\phantom{-}0.6$  & $\phantom{-}0.842$ $(0.353)$ & $\phantom{-}0.262$ $(0.364)$ & $\phantom{-}0.564$ $(0.188)$ \\
			$m_1$               & $-0.1$            & $-0.111$ $(0.128)$           & $-0.114$ $(0.015)$           & $-0.100$ $(0.000)$           \\
			$m_2$               & $\phantom{-}0.3$  & $\phantom{-}0.351$ $(0.164)$ & $\phantom{-}0.288$ $(0.014)$ & $\phantom{-}0.300$ $(0.001)$ \\ \hline
			KS                  &                   & $\phantom{-}0.326$           & $\phantom{-}0.222$           & $\phantom{-}0.078$           \\ \hline
		\end{tabular}
	\end{center}
	\end{adjustwidth}
	\caption{Expected value of estimates and RMSE (in parentheses) for the $WVAG $ model fitted to simulated data with $c=0.1$.}
	\label{v-alpha-g-table4}
\end{table}

\subsection{Real data results}\label{realdatasec}
Next, we fit the $WVAG $ and $VAG $ models to the S\&P500 and FTSE100 indices as the bivariate price process \eqref{priceproc} for a 5 year period from 14 February 2011 to 12 February 2016 with daily closing price observations taking $c=1$. The estimated parameters, goodness of fit statistics and standard errors computed using 100 bootstrap samples are listed in Table \ref{v-alpha-g-table}. Contour plots of the fitted distributions and scatter plots of the bivariate log returns are shown in Figure \ref{v-alpha-g-plot}.

\begin{sidewaystable}[!htbp]
	\begin{center}
		\begin{tabular}{ccccccc}
			\hline
			& \multicolumn{2}{c}{\bf MOM} & \multicolumn{2}{c}{\bf MLE} & \multicolumn{2}{c}{\bf DME}\\ \hline
			{\bf Parameter} & {\bf $WVAG $} & {\bf $VAG $} & {\bf $WVAG $} & {\bf $VAG $} & {\bf $WVAG $} & {\bf $VAG $} \\ \hline
			$a$                 & $\phantom{-}0.695$ $(0.199)$ & $\phantom{-}0.696$ $(0.217)$ & $\phantom{-}0.962$ $(0.139)$ & $\phantom{-}1.087$ $(0.117)$ & $\phantom{-}0.899$ $(0.138)$ & $\phantom{-}1.114$ $(0.327)$ \\
			$\alpha_1$          & $\phantom{-}1.436$ $(0.398)$ & $\phantom{-}1.436$ $(0.349)$ & $\phantom{-}0.919$ $(0.105)$ & $\phantom{-}0.908$ $(0.104)$ & \multicolumn{2}{c}{$\phantom{-}0.898$ $(0.143)$} \\
			$\alpha_2$          & $\phantom{-}0.933$ $(0.277)$ & $\phantom{-}0.707$ $(0.120)$ & $\phantom{-}0.837$ $(0.105)$ & $\phantom{-}0.919$ $(0.085)$ & \multicolumn{2}{c}{$\phantom{-}0.878$ $(0.132)$} \\
			$1000\mu_1$         & $-1.153$ $(0.704)$           & $-1.156$ $(0.595)$           & $-0.403$ $(0.550)$           & $-0.630$ $(0.694)$           & \multicolumn{2}{c}{$-0.562$ $(0.563)$} \\
			$1000\mu_2$         & $-1.246$ $(0.704)$           & $-1.255$ $(0.789)$           & $-0.891$ $(0.586)$           & $-0.867$ $(0.578)$           & \multicolumn{2}{c}{$-1.166$ $(0.554)$} \\
			$10000\Sigma_{11}$  & $\phantom{-}0.982$ $(0.070)$ & $\phantom{-}0.980$ $(0.069)$ & $\phantom{-}0.976$ $(0.061)$ & $\phantom{-}0.935$ $(0.064)$ & \multicolumn{2}{c}{$\phantom{-}0.928$ $(0.071)$} \\
			$10000\Sigma_{22}$  & $\phantom{-}1.006$ $(0.053)$ & $\phantom{-}1.007$ $(0.062)$ & $\phantom{-}1.028$ $(0.063)$ & $\phantom{-}1.014$ $(0.060)$ & \multicolumn{2}{c}{$\phantom{-}1.051$ $(0.086)$} \\
			$10000\Sigma_{12}$  & $\phantom{-}0.994$ $(0.072)$ & $\phantom{-}-$               & $\phantom{-}0.813$ $(0.089)$ & $\phantom{-}-$               & $\phantom{-}0.844$ $(0.095)$ & $\phantom{-}-$ \\
			$1000m_1$           & $\phantom{-}1.422$ $(0.604)$ & $\phantom{-}1.426$ $(0.519)$ & $\phantom{-}0.705$ $(0.430)$ & $\phantom{-}0.879$ $(0.575)$ & \multicolumn{2}{c}{$\phantom{-}0.982$ $(0.444)$} \\
			$1000m_2$           & $\phantom{-}1.198$ $(0.617)$ & $\phantom{-}1.207$ $(0.722)$ & $\phantom{-}0.847$ $(0.440)$ & $\phantom{-}0.882$ $(0.470)$ & \multicolumn{2}{c}{$\phantom{-}1.066$ $(0.410)$} \\ \hline
			$-\ln L$            & $-8465.969$                  & $-8192.168$                  & $-8496.315$                  & $-8239.301$                  & $-8492.798$                  & $-8237.767$ \\
			$\chi^2$            & $\phantom{-}144.034$         & $\phantom{-}702.241$         & $\phantom{-}118.574$         & $\phantom{-}586.789$         & $\phantom{-}99.359$          & $\phantom{-}584.547$ \\
			KS                  & $\phantom{-}0.073$           & $\phantom{-}0.151$           & $\phantom{-}0.050$           & $\phantom{-}0.139$           & $\phantom{-}0.048$           & $\phantom{-}0.138$ \\
			\hline
		\end{tabular}
		\caption{Parameter estimates and standard errors (in parentheses) when fitted to the S\&P500-FTSE100 data set. Estimation for the marginal parameters of $WVAG $ and $VAG $ models are identical.}
		\label{v-alpha-g-table}
	\end{center}
\end{sidewaystable}

\begin{figure}[!htbp]
	\begin{center}
		\includegraphics{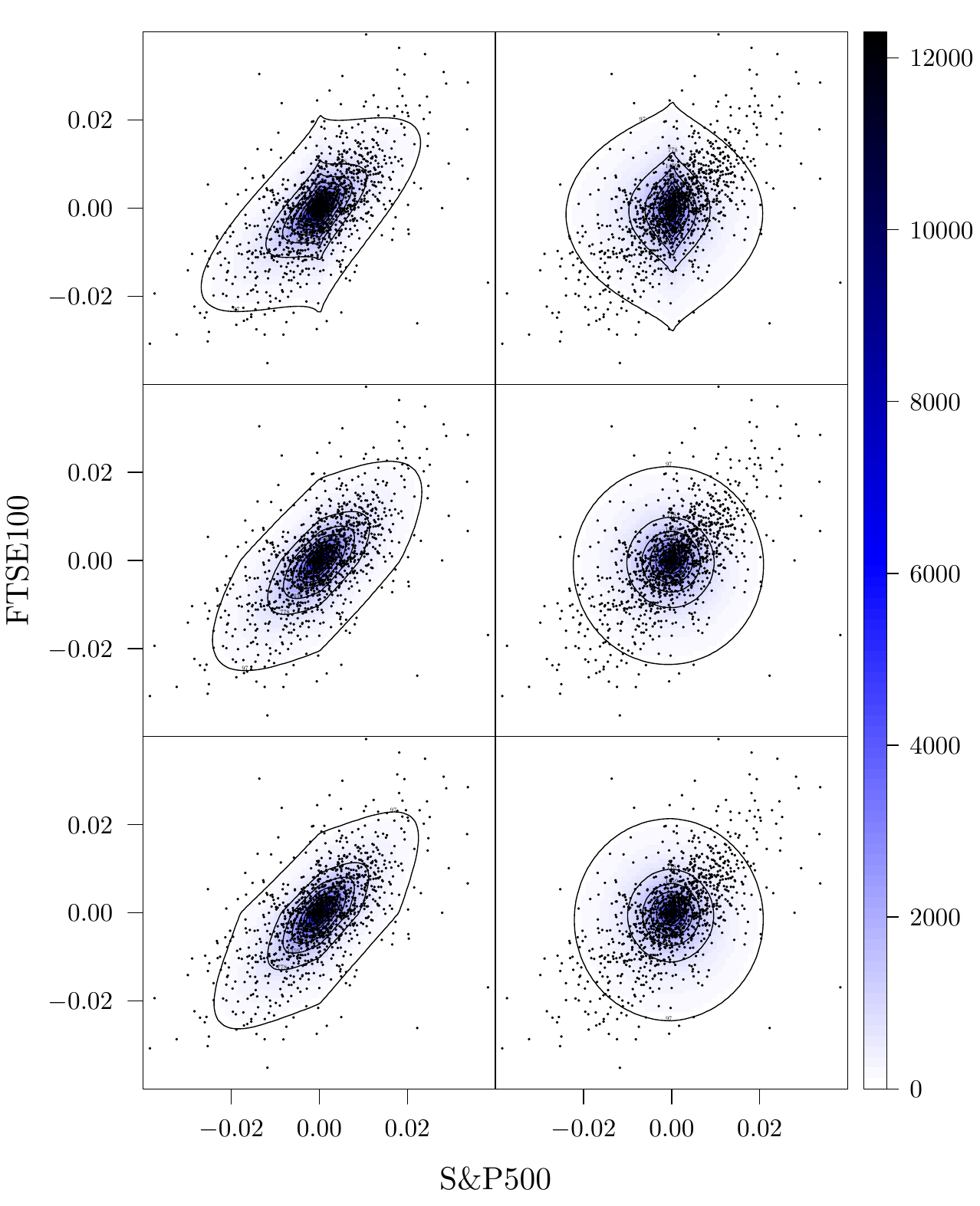}
		\caption{Scatterplots of log returns of the S\&P500-FTSE100 data set and contour plots of the fitted distributions using the weak $VAG $ model (left) and $VAG $ model (right) with MOM (top), MLE (middle), DME (bottom).}
		\label{v-alpha-g-plot}
	\end{center}
\end{figure}

Note that the Fourier invertibility condition is satisfied for all fitted models. Based on the $\chi^2$, KS statistic and contour plots, the $WVAG $ model produces a better fit than the $VAG $ model. In addition, for the $WVAG $ model, DME gives a fit with a lower $\chi^2$ and KS statistic than MLE and MOM.

Assuming that the log returns satisfies the $WVAG $ model, a likelihood ratio test can be used to test the hypothesis $H_0:\Sigma_{12}=0$ versus $H_1:\Sigma_{12}\neq0$.
The test statistic $D=514.03$ is asymptotically $\chi^2$ distributed with 1 degree of freedom. The $p$-value is $<10^{-4}$, so the $VAG $ model is rejected.
Indeed, the $VAG $ model is not suited for modelling strong correlation since $\operatorname{Cov}((B\circ T)_1(1),\allowbreak (B\circ T)_2(1))=a\alpha_1\alpha_2\mu_1\mu_2$ by Proposition \ref{wvagprop} (vi), which is approximately 0 when $\mu_1\mu_2$ is.

It has been suggested that log-returns should be self-decomposable \cite{Bi06,BiKi02,CGMY07}. Note that $\bfmu=(-0.0004,-0.0008)$ is very close to ${\bf 0}$, which suggests that the log-returns process $Y$ is indeed self-decomposable (see Proposition \ref{wvagprop} (vii)). A likelihood ratio test can be used to test this hypothesis, $H_0:\bfmu={\bf 0}$ versus $H_1:\bfmu\neq{\bf 0}$. The test statistic $D=4.11$ is asymptotically $\chi^2$ distributed with 2 degrees of freedom. The $p$-value is 0.128, so at a 5\% significance level we cannot reject that $Y$ is self-decomposable.

\section{Conclusion}

The $WVAG $ process constructed by using weak subordination generalises the $VAG $ process, and we obtain a condition for Fourier invertibility in Theorem \ref{fourierinvert}. We have shown that MOM, MLE and DME can be used to estimate the parameters of a $WVAG $ process, and find that in our simulations MLE produces a better fit when the Fourier invertibility condition holds, while DME produces a better fit when it is violated. However, MLE may still produce good parameter estimates even when the Fourier invertibility condition is violated. In all cases, MOM produces the worst fit. We find that the $WVAG $ process exhibits a wider range of dependence and produces a significantly better fit than the $VAG $ process when used to model the S\&P500-FTSE100 data set, and that DME produces the best fit in this situation.

\appendix
\section{Appendix}
\noindent{\bf L\'evy process.}~The reader is referred to the monographs~\cite{Ap09,b,s} for necessary material on L\'e\-vy pro\-cesses, to \cite{BB14,CS09,CT} for financial applications, while
our notation follows~\cite{BKMS16,BLMa}. For $A\subseteq\RR^n$, let $A_*:=A\backslash\{{\bf 0}\}$ and let $\eins_A(\omega)$ denote the indicator function. Let $\DD:=\{\bfx\in\RR^n:\|\bfx\|\le 1\}$ be the Euclidean unit ball centred at the origin. The law of an $n$-dimensional L\'evy process $X=(X_1,\dots,X_n)=(X(t))_{t\ge 0}$ is determined by its characteristic function $\Phi_X:=\Phi_{X(1)}$, with
\[\Phi_{X(t)}(\bftheta)\,:=\,\EE\exp(\rmi\skal\bftheta{X(t)})\,=\,\exp(t\allowbreak\Psi_X(\bftheta))\,, \qquad t\!\ge\!0,\]
and L\'evy exponent $\Psi_X:=\Psi$, where
\begin{equation*}\label{0.1}
\Psi(\bftheta)\,:=\,
\rmi \skal {\bfmu}\bftheta\!-\!\frac 12\;\|\bftheta\|^2_{\Sigma}
+\int_{\RR^n_*}\left(e^{\rmi\skal\bftheta \bfx}\!-\!1\!-\!\rmi\skal\bftheta \bfx \eins_\DD(\bfx)\right)\,\XXX(\rmd \bfx)\,,
\end{equation*}
$\bftheta\in\RR^n$, $\bfmu\in\RR^n$, $\Sigma\in\RR^{n\times n}$ is a covariance matrix,
and $\XXX$ is a nonnegative Borel measure on $\RR^n_*$ such that
$\int_{\RR^n_*}(\|\bfx\|^2\wedge 1)\,\XXX(\rmd \bfx)<\infty$. We write $X\sim L^n(\bfmu,\Sigma,\XXX)$ provided $X$ is an $n$-dimensional L\'evy process with canonical triplet $(\bfmu,\Sigma,\XXX)$.

A subordinator $T\!\sim\!S^n(\TTT)\!=\!L^n(\bfmu,0,\TTT)$ is drift-less if its drift $\bfmu\allowbreak\!-\!\int_{\DD_*} \bft\,\TTT(\rmd \bft)={\bf 0}$. All subordinators considered in this paper are drift-less.

An $n$-dimensional random variable $X$ is {\em self-decomposable} if for any $0\!<\!b\!<\!1$, there exists a random variable $Z_b$, independent of $X$, such that $X\eqd bX+Z_b$. A L\'evy process $X$ is self-decomposable if $X(1)$ is.\\[1mm]
\noindent{\bf Strongly subordinated Brownian motion.} Let $B=(B_1,\dots,B_n)\sim BM^n(\bfmu,\Sigma)$ be a Brownian motion and $T\!=\!(T_1,\dots,T_n)\!\sim\!S^n(\TTT)$ be a drift-less subordinator. A process $B\circ T$ is the traditional or {\em strong subordination of $X$ and $T$} if $(B\circ T)(t):=(B_1(T_1(t)),\dots,B_n(T_n(t)))$, $t\ge 0$.\\[1mm]
\noindent{\bf Weakly subordinated Brownian motion.}~Let $\bft\!=\!(t_1,\dots t_n)\!\in\![0,\infty)^n$, $\bfmu\!=\!(\mu_1,\dots \mu_n)\!\in\!\RR^n$ and $\Sigma\!=\!(\Sigma_{kl})\!\in\!\RR^{n\times n}$ be a covariance matrix. Introduce the outer products~$\bft\tr\bfmu\in\RR^n$ and~$\bft\tr\Sigma\!\in\!\RR^{n\times n}$ by
\begin{equation}\label{outproddefn}
\bft\tr\bfmu:=(t_1\mu_1,\dots,t_n\mu_n)\quad\text{and}\quad(\bft\tr \Sigma)_{kl}:=\Sigma_{kl}(t_k\wedge t_l),\quad 1\le k,l\le\!n.
\end{equation}
Let $B\!\sim\!BM^n(\bfmu,\Sigma)$ be an $n$-dimensional Brownian motion and $T\!\sim\!S^n(\TTT)$ be an $n$-dimensional drift-less subordinator. A L\'evy process $B\odot T$ is called the
{\em weak subordination of $B$ and $T$}~(see \cite{BLMa}, their Proposition~3.1) if it has L\'evy exponent
\begin{align}\label{wsexp}
\Psi_{B\odot T}(\bftheta)=\int_{[0,\infty)^n_*}\left(\exp\left( \rmi\skal\bftheta{\bft\tr\bfmu}-\frac 12\|\bftheta\|^2_{\bft\tr\Sigma}\right)-1\right)\,\TTT(\rmd\bft)\,,
\end{align}
$\bftheta\in\RR^n$. Note that a more general definition of weak subordination and a proof of existence is given in \cite{BLMa}.

Assume that independent $B$ and $T$. If $T$ has indistinguishable components or $B$ has independent components, then 
\begin{align}\label{strongequalsweak}
B\circ T\eqd B\odot T
\end{align}
Otherwise $B\circ T$ may not a L\'evy process, but $B\odot T$ always is (see~\cite{BLMa}, their Proposition~3.3 and 3.9).

\section*{Acknowledgement}
B. Buchmann's research was supported by ARC grant DP160104737. K. Lu's research was supported by an Australian Government Research Training Program Scholarship.

\end{document}